\documentclass{informx}
\usepackage{graphics}%
\usepackage[dvips]{graphicx}
\usepackage{clrscode}
\usepackage{booktabs}
\usepackage{array}
\usepackage{algorithm}
\usepackage{algorithmic}

\begin{document}
\begin{frontmatter}                           
%

\title{A Fault-Tolerant Emergency-Aware Access Control Scheme \\
 for Cyber-Physical Systems}
\author{\textbf{\fnms{Guowei} \snm{Wu}}},
\author{\textbf{\fnms{Dongze} \snm{Lu}}},
\author{\textbf{\fnms{Feng} \snm{Xia}}\thanks{Corresponding author; e-mail: f.xia@ieee.org.}},
\author{\textbf{\fnms{Lin} \snm{Yao}}}
\address{School of Software, Dalian University of Technology, \\
Dalian 116620, China}

\runningauthor{G.W. Wu, D.Z. Lu, F. Xia, L. Yao}
\runningtitle{A Fault-Tolerant Emergency-Aware Access Control Scheme for Cyber-Physical Systems}
%
\maketitle
\begin{abstract}
Access control is an issue of paramount importance in cyber-physical systems (CPS). In this paper, an access control scheme, namely FEAC, is presented for CPS. FEAC can not only provide the ability to control access to data in normal situations, but also adaptively assign emergency-role and permissions to specific subjects and inform subjects without explicit access requests to handle emergency situations in a proactive manner. In FEAC, emergency-group and emergency-dependency are introduced. Emergencies are processed in sequence within the group and in parallel among groups. A priority and dependency model called PD-AGM is used to select optimal response-action execution path aiming to eliminate all emergencies that occurred within the system. Fault-tolerant access control polices are used to address failure in emergency management. A case study of the hospital medical care application shows the effectiveness of FEAC.
\end{abstract}
\begin{keyword}
fault-tolerance, access control, cyber-physical systems, emergency management.
\end{keyword}
\end{frontmatter}

\section{Introduction}\label{intro}

Cyber-Physical Systems (CPS) is the integration of computing, communication and storage capabilities with monitoring and controlling the entities in the physical world \cite{1,2,3,4}. The emergence of such systems has effect on the revolution including high confidence medical devices and systems, assisted living, traffic control and safety, advanced automotive systems, process control, energy conservation, environmental control, avionics, instrumentation, critical infrastructure control, distributed robotics, defense systems, manufacturing, and smart structures \cite{5,6,7,8,9}. The security issues are crucial for CPS applications because the entities within the systems not only interact with each other, but also with the physical environment, thus the security issues must be addressed before CPS applications could be widely deployed. Access control is an essential component of CPS security to protect sensitive resources and services from unauthorized access and qualify the behavior of entities within the system.

Existing access control schemes as RBAC, GRBAC, CAAC \cite{10,11,12} are traditionally provide access services in a passive manner, which need the subject explicitly require the access. The access control polices of these schemes are statically defined before the application deployed, and cannot be adjusted according to the change of system environment dynamically. Especially in emergency situation, traditional access control schemes cannot provide proper privileges to execute the response actions to avoid the failure of the system. In CPS application, physical environment is an import part of the whole system, when making the access control decisions, the environment context and the whole system context (not only the context of access subject, but also the object context and system states) must be taken into account.

In this paper, a new access control scheme called FEAC (Fault-tolerant Emergency-aware Access Control) is proposed, which provides a proactive and adaptive access control polices especially to address multiple emergencies management problem and supposes the fault-tolerant scheme for CPS applications. PD-AGM (Priority and Dependency-Action Generation Model) is introduced to select the optimal response action path for eliminating all the active emergencies within the system, reference the methods in \cite{13} and \cite{14}. The priority and dependency relationships of emergencies are used to exclude the infeasible response action paths and relieve the emergencies combination state explosion problem. In order to handle all the emergencies timely, emergency-group and emergency-role are introduced for parallelly processing multiple emergencies.

The remainder of this paper is organized as follows. Section~\ref{rel} introduces related work. Section~\ref{pre} presents the primary concepts of emergency management. Section~\ref{FEAC} presents the FEAC scheme including PD-Action Generation Model and access control policy. Section~\ref{val} gives the validation proof of the system model. Section~\ref{cs} presents a case study to demonstrate the access control scheme. Section~\ref{con} concludes the paper.

\section{Related Work}\label{rel}
With the growth of pervasive computing technologies, researchers pay increasing attention to a new area named CPS. Security is one of the most important problems that must be addressed in CPS applications \cite{15,16}. Much work has been done with respect to access control for pervasive computing systems and other systems, which is a primary component of system security.

Role Based Access Control (RBAC) \cite{17} is one of the most influential schemes for authorized access information. Within an organization, roles are created for various job functions. The permissions to perform certain operations are assigned to specific roles without directly associated with subjects. RBAC provides an effective and easy way to enforce complex access control policies. Different from RBAC, Context Based Access Control (CBAC) \cite{18} avoids the notion of roles and directly associates permissions to the subjects by the context information. Usage Control (UCON) \cite{19,20} combines the notions of access control, trust management and digital rights management to provide fine-grained access control to unknown subjects. None of these schemes have the ability of privacy preservation when the system under emergency situations. The access control polices runs in a reactive manner and the explicit access require from subjects is needed. The access control policies of RBAC are static in nature and predefined before deployment. Though UCON and CBAC have the ability to change the permissions available to subjects, they only consider the change of the subject context, which is too simplistic for CPS to manage the emergency situations.

The Policy Spaces (PS) model  \cite{21} provides adaptive emergency management. It divides polices into groups and provides access privilege for specific situations. However, the PS works in reactive manner in nature and cannot control the emergencies within the system timely. Criticality-Oriented Access Control (COAC) \cite{22} firstly introduces the notion of altering access control privileges to enable emergency management for smart-spaces. The alternate idea is to promote the role of specific subjects in the space to execute response actions in a limited duration. COAC can only control the systems with single emergency. Criticality aware Access control (CAAC) \cite{23,24} expands the COAC scheme with a stochastic modeling framework for evaluating the management of multiple emergencies in smart-spaces. The stochastic modeling framework provides a mechanism for determining the response actions or dealing with the stochastic nature of emergencies. CAAC scheme presents a more proactive and adaptive manner than other schemes.

Although the existing schemes play important roles to guarantee the security requirement of CPS applications in various degrees, designing a proactive and adaptive access control is still a challenging issue in CPS. In this paper, a new access control especially to address the emergency management problem and fault-tolerant problem for CPS applications is proposed based on other relevant schemes. The major differences between this work and the aforementioned schemes are as follows:
\renewcommand{\labelenumi}{(\arabic{enumi})}
\begin{enumerate}
  \item PD-AGM model is introduced to select the optimal response action path for eliminating emergencies. The priority and dependency relationships of emergencies are used to exclude the infeasible response action paths and relieve the emergency combination state explosion problem.
  \item The Influence-factor is employed to precisely represent the influence between emergencies in the group, which can help to generate the optimal response action path. One emergency can influence other emergencies in terms of the emergency-duration, the probability of success and the execution time for specific response actions.
  \item To guarantee the emergencies be processed timely, the emergency-group and emergency-role are adopted for parallelly processing the emergencies within the system. Emergencies are grouped by the entity they belong to. The emergencies in different groups are parallelly processed, while the emergencies in the group are processed in sequence.
  \item Role-mapping and constraints are proposed for selecting proper subjects to execute specific response actions. A hierarchical role structure is used for selecting the most suitable subject.
  \item Fault-tolerant scheme is proposed for protecting the normal running of the system by proactive alternating permissions and services to substitute entities after the failure of emergency management.
\end{enumerate}

\section{Preliminaries}\label{pre}
In this section, we introduce some of the principal concepts of emergency management used in FEAC.

\subsection{Emergency}\label{Emy}
Emergency is defined as the effect of series of events in physical world \cite{25}, which can cause the system goes into unstable states. These events are called emergency events. The time duration for executing some operations to restoring the system back to normal state is called emergency-duration (Ed), and the corresponding operations are called response actions for specific emergency within the system \cite{26}. In medical emergency, the time duration is called golden-hour, during which the prompt medical treatment can save life. If the emergency-duration has expired, the system will fail and cause the loss of property and even life.

\subsection{Emergency Management}\label{Emyman}
Emergency management is designed to control all the active emergencies that occur in the system, to protect the sensitive information and to limit the actions of the entities under emergency situations. Emergency management enables response actions to eliminate the active emergencies to avoid system failures in a proactive manner. Emergency management for single emergency includes three phases which are shown in Figure~\ref{fig1}: Detection and Preparation, Response, and Post-process. In the Detection and Preparation phase, the emergency is detected in a timely manner in order to save time for executing the response actions. The response action path is selected after the detection of the emergency. Then in the Response phase, the responses actions are performed to eliminate the active emergencies. Emergency-role and permissions for executing response actions are assigned to the selected subjects. Finally in the Post-process phase, the response actions are evaluated, and the properties of the emergency may be updated if necessary. The response action path and the selection of the response actions might be changed for the next execution.

In CPS, multiple emergencies may occur at one time. Unlike single emergency management, multiple emergencies management is more complex:
 \begin{enumerate}
   \item The system not only needs to track the execution of the emergencies already existing, but also needs to detect the occurrence of new emergencies.
   \item Multiple emergencies that occurred on the same entity must be processed in sequence, and only one emergency could be processed at one time. The schedule for the process of multiple emergencies should be considered.
   \item The execution sequence is influenced by the dependency relationships between different emergencies. Emergency-dependency exists between emergencies that occurred on the same entity, and also between emergencies that occurred on different entities.
   \item Parallel process of multiple emergencies is crucial for the performance of the emergency management.
 \end{enumerate}
\begin{figure}
   \centerline{\includegraphics[width=2.7in]{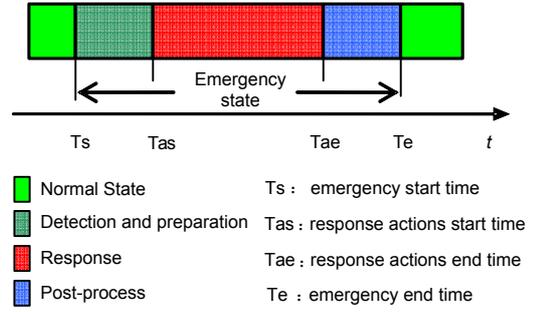}}
   \caption{Single Emergency Management}\label{fig1}
\end{figure}
\subsection{Emergency-group and Emergency-dependency}\label{Emy}

 According to the entity the emergency belongs to, multiple emergencies are divided into different groups, named emergency-group. The emergencies in different groups can be parallelly processed. In FEAC, a specific environment emergency-group is defined to specify the emergencies caused by the environment. The environment emergency affects other emergencies in various ways. The affected emergencies can be processed only when the environment emergency is eliminated. Fire in the intensive care unit (ICU) is an example of environment emergency. In this context, other rescue actions for patients can be deployed only when the fire under control.

 Emergency-dependency is first introduced in this paper. Emergency-dependency can reduce the states of action generation model. In FEAC, the following five kinds of dependency relationship are considered.

\begin{definition}\label{D1}
\textbf{Entity-dependency}. Emergencies are divided into different emergency-groups according to the entity each emergency belongs to.
\end{definition}

\begin{definition}\label{D2}
\textbf{Time-dependency}. This dependency relationship indicates the processing sequence between the emergencies within an emergency-group.
\end{definition}

\begin{definition}\label{D3}
\textbf{Environment-dependency}. This is a special dependency property of CPS applications. In CPS, the computing system interacts with physical environment. Consequently, environment emergency affects other entity emergencies. The dependency relationship between entity emergencies and the environment emergency is called environment-dependency.
\end{definition}

\begin{definition}\label{D4}
\textbf{Resource-dependency}. The resource-dependency relationship indicates the competition of certain resource. The emergencies in different emergency-groups may have resource-dependency relationship.
\end{definition}

\begin{definition}\label{D5}
\textbf{Subject-dependency}. Some of the response actions with respect to different emergencies need the same subject to perform. One subject cannot perform multiple response actions at one time.
\end{definition}

 To avoid dependency deadlock circle within a group, we assign priorities to different emergencies, and assume the time-dependency may exist only from the emergency with higher priority to that with lower priority. Thus the dependency deadlock circle can be broken. The time-dependency affects the execution sequence in the group. The entity-dependency determines the emergency-group of the emergencies within the system. The last three dependencies exist between groups, and affect parallel process of emergencies. Subject-dependency conflict can be solved by reselecting subjects for the emergencies.

\section{Fault-tolerant Emergency-aware Access Control}\label{FEAC}
 Figure~\ref{fig2} depicts the schematic diagram of the FEAC scheme for CPS applications. FEAC runs under the normal state, emergency state and fault-tolerant state. The transition happens with the direction of arrow. When the emergency event occurs, the system moves from normal state to emergency state. When no emergency is active, the state moves back to normal state. The failure in finding the optimal response path in emergency management brings the system to fault-tolerant state. If the failure is addressed successfully, the system turns back to emergency state and continues emergency processing; otherwise, the system goes to disaster state and disables all services.

 \begin{figure}
   \centerline{\includegraphics[width=2.7in]{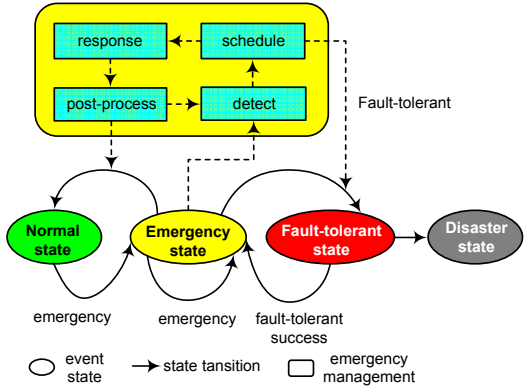}}
   \caption{Illustration of FEAC}\label{fig2}
\end{figure}

 The constitution of FEAC structure is described in Figure~\ref{fig3}. The lower layer is the data layer, which provides the data information for decision entity to generate the access control decisions. The original data include subjects (S), objects (O), permissions (P), roles (R) and constraints (C) in core RBAC. Additional data include context and constraint. Upper parts in data layer are abstract data and dynamic data. ACMD (Access Control Meta-Data) abstracts the relations between subjects, roles and corresponding permissions, and provides meta-data for other units. CMU (Constraints Management Unit) provides access control constraints for both normal state and emergency state. DCMU (Dynamic Context Management Unit) processes the context data that are collected in the lower parts and provides higher level contextual information for the components of decision layer. In decision layer, EMU (Emergency Management Unit) and FTU (Fault-Tolerant Unit) utilize the data information provided by data layer, and interact with RMU to handle the emergencies that occur within the system and generate appropriate access control permissions. AMU (Account Management Unit) records all the action events of the system. ACPM (Access Control Policy Management) performs the access control.

\begin{figure}
   \centerline{\includegraphics[width=2.7in]{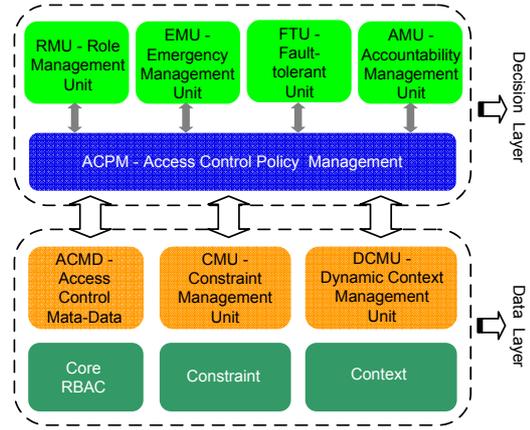}}
   \caption{FEAC System Structure}\label{fig3}
\end{figure}

 Initially, the administrator of the system establishes the set of roles, permissions and constraints when the FEAC is deployed. Under normal state, the permissions in the objects' ACLs are used to make access decisions. If necessary, entries are dynamically added into the ACL according to the constraints and contexts of subjects and objects. When the system is in emergency state, FEAC uses PD-AGM to evaluate the characteristics of the emergencies, select the optimal response action path and response actions and proactively activate the permissions. The selected subjects are proactively informed to access the system with the permissions assigned to the emergency-role.

\subsection{Emergency Management in FEAC}\label{Emy-Mana}
 The emergency management structure is shown in Figure~\ref{fig4}.  ACMI (Access Control Meta-data Interpreter) and CCI (Constraint and Context Interpreter) are the interfaces of meta-data, constraints and contexts information in data layer. Emergency event is first detected by EDU (Emergency Detection Unit), and then EPDU (Emergency Property Determination Unit) determines the properties, such as Ed and TS (Task set) of the emergency. Most of the operations are performed by EPU (Emergency Process Unit), which uses the data information and emergency properties to schedule emergencies and execute the response actions. The principal component of EPU is PD-AGM, which is responsible for generating the optimal response action path and the response action for the emergency. EPU interacts with other system components through EMIU (Emergency Management Interaction Unit). ENU (Emergency Notification Unit) informs the selected subjects of the permissions they have been assigned. Each permission is associated with a time duration to restrict the time period that the subject can process the emergency.

\begin{figure}
   \centerline{\includegraphics[width=2.7in]{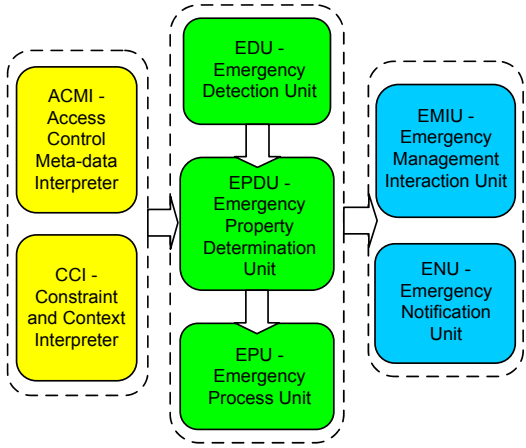}}
   \caption{PD-Action Generation Model}\label{fig4}
\end{figure}

\subsubsection{ PD-Action Generation Model}\label{PDAGM}

 Action Generation Model (AGM) is first introduced in \cite{27}. It is an effective way to determine the response actions for emergencies. In AGM, Response Links (RL) represents the transitions that make the system restore to the normal state. Stochastic crisis planning technique developed in \cite{28} is used in AGM to model all the possible emergency states. Due to the diversity of the emergency combinations, the state explosion problem remains an open issue.

 PD-AGM is an extension of AGM and FDs\cite{29}, which uses priority and emergency dependency for response action path generation. PD-AGM can reduce the number of emergency states of the system and relieve state explosion problem. Figure 5 illustrates the generation of the emergency state transition graph. The solid line and the longer dotted line distinguish the priority hierarchy and dependency hierarchy. The shorter dotted line makes a distinction between different priorities in the dependency hierarchy. If the emergencies have the same priority, the model uses stochastic method to generate all the possible combinations. Each edge (RL) is associated with time, TS, Ed and execution time of specific emergency. The algorithm for state transition is described in Algorithm~\ref{alg1}.

    \begin{algorithm}
    \caption{\small{Emergency state transition graph generation}}
    \label{alg1}
    \begin{algorithmic}[1]
    \REQUIRE~~\\
        Group of emergencies $g$;
    \ENSURE~~\\
        Root node of emergency states transition graph $t$
    \STATE Sorted the emergencies in group g by priority
    \FORALL{emergency priority hierarchy}
        \STATE Put dependency emergencies into the dependency hierarchy
        \STATE Remove them from corresponding priority hierarchy
        \WHILE{dependency hierarchy is not Null}
        \STATE Put dependency emergencies into higher hierarchy
        \STATE Remove them from corresponding priority hierarchy
        \ENDWHILE
    \ENDFOR
    \FORALL{priority-dependency hierarchy}
        \STATE Sort the emergencies by priority
    \ENDFOR
    \STATE Set the root node $t$ of the emergency states
    \FORALL{sorted priority hierarchy}
        \STATE Randomly generate state transition path, and add it to the graph $t$
        \STATE Calculate and set the edge properties as Ed and execute time
        \STATE Set the tail node of this hierarchy.
    \ENDFOR
    \STATE Set the tail node as normal state and return the root node $t$
    \end{algorithmic}
    \end{algorithm}

\subsubsection{RL path and Response Actions}\label{rlra}

The choice of the optimal RL path and response actions for the group of emergencies is crucial for the emergency management. $TS$ presents the set of response actions which are used to mitigate the emergency, denoted as $<\{a_1, a_2, ... , a_k\}, t, p>$, where $k$ is the number of sub-actions, $t$ is the execution time of the set of response actions and $p$ is the probability of successful execution of the response actions. In this paper, we select TS that has the highest probability of success and associate it with the corresponding RL from the set of TSs. The RL is denoted as $<Eid, TS, Ed>$, where $Eid$ is the identity of the emergency to be processed and $Ed$ is the Emergency duration for executing the set of response actions in TS. Only when all groups of emergencies have been eliminated, the system is considered to recover back to normal state. The TS is affected by the probability $p$ and the available resources. We choose the TS with the highest probability in the executable set of TSs.

\begin{figure}
   \centerline{\includegraphics[width=2.7in]{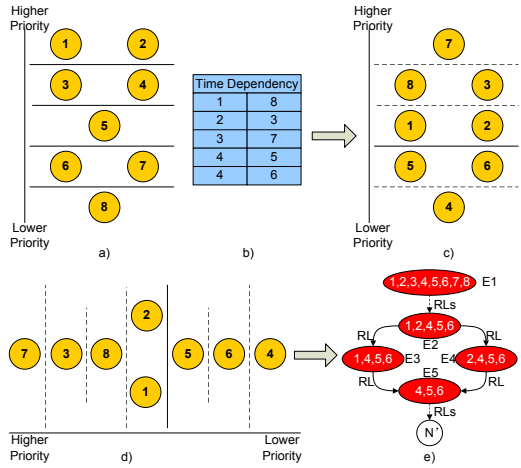}}
   \caption{State Transition of PD-AGM}\label{fig5}
\end{figure}

The choice of RL depends on its $P$-$value$. The definition of the $P$-$value$ is as follows:

\begin{definition}\label{D6}
\textbf{P-value}. $P$-$value$ is defined as the total probability of successful recovery to the normal state from current emergency state.
\end{definition}

The $P$-$value$ $Pv(i)$ of current node $i$  is calculated by (\ref{equ1}) in a recursive manner:

\begin{equation}
y=\left\{\begin{array}{l@{~~~~}l}\label{equ1}
0 & any~Ed~expired \\
1 & normal~state \\
\max_{j=1...k}(p(i,j)\cdot Pv(j)) & childhood~node \\
\end{array}\right.
\end{equation}

Here $p(i,j)$ is the probability of reaching emergency state $j$ from state $i$. If the result of the $P$-$value$ on root node is not equal to zero, the path with the largest $P$-$value$ will be chosen as the optimal response path. If more than one path has the same $P$-$value$, the path with the shortest execution time will be chosen. If the $P$-$value$ is equal to zero, it indicates that all the RLs cannot meet the time requirement of emergency duration. Fault-tolerance is used to overcome this problem by choosing substitution entity. After fault-tolerance is performed, two heuristic selection algorithms are used to minimize the lose of property and life: one algorithm chooses the RL with the maximum probability of successful transition to normal state and the other chooses the RL with the minimum execution time as the optimal response action path.

The emergencies are not independent in the group, and they interact with each other. The Influence Factor $\sigma$ is used to represent the influence between the emergencies. The corresponding definition is as follows:

\begin{definition}\label{D7}
\textbf{Influence Factor $\sigma$}. Influence factor is used to abstract the influence on the emergency by other active emergencies in the group. Three properties are affected by $\sigma$: priority, execution time and Ed, which can be calculated with the following equations:

\begin{equation}
p'_{ji}=(1-\sigma)p_{ji}\label{equ2}
\end{equation}
\begin{equation}
t'_{ji}=(1+\alpha \cdot \sigma)t_{ji}\label{equ3}
\end{equation}
\begin{equation}
Ed'_{ji}=(1-\beta \cdot \sigma)Ed_{ji}\label{equ4}
\end{equation}

\end{definition}

where $\alpha$ and $\beta$ are the coefficients for execution time and emergency duration respectively.

\subsection{Fault-tolerance in FEAC}\label{ft}

The fault-tolerance module always runs in a passive manner when the system crashed. In our scheme, it executes in a proactive way by predicting failure of the system. As mentioned above, when the $P$-$value$ is 0, it is implied that the emergency management fails to eliminate all the active emergencies in the emergency-group. After the mitigation actions have been processed, the system goes to disaster state. FEAC employs another hierarchy to protect the system using a fault-tolerant method.

\begin{figure}
   \centerline{\includegraphics[width=2.7in]{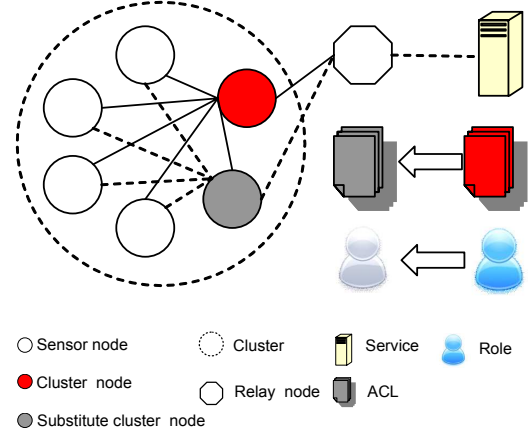}}
   \caption{Example of fault-tolerance in FEAC}\label{fig6}
\end{figure}

Figure~\ref{fig6} gives an example of exchanging cluster head node of wireless sensor network in order to tolerate the damage of the cluster head node. The scheme selects the proper node instead of the previous one, adds the entities in previous ACL to the substitute node's ACL and assigns the roles of previous one to the substitute one. The concrete algorithm for achieving fault tolerance is illustrated in Algorithm~\ref{alg2}. The algorithm checks whether the entity could tolerate faults. When the result is false, and if there does not exist any substitute entity, the system moves to disaster state and fails; otherwise, a substitute entity will be activated to continue performing the jobs instead of the original entity. The substitute entity is selected from the EFGT (Entity Function Group Table), which has the same function with the entity encountering the fault.


    \begin{algorithm}
    \caption{Fault-tolerance algorithm}
    \label{alg2}
    \begin{algorithmic}[1]
    \REQUIRE~~\\
        entity $e$ that needs to tolerate the fault;
    \IF{the property that indicates feasibility of fault-tolerance is false}
        \STATE exits and moves to $\textbf{Disaster}$ state
    \ELSIF{$se = findSubEntity(e)$ is Null}
        \RETURN exits and moves to $\textbf{Disaster}$ state
    \ELSE
        \IF{$getACL(e)$ is not Null}
            \FORALL{entries in the ACL}
                \STATE Add it to the ACL of $se$
                \STATE Inform the subject that has activated the role of the entity for the exchange of permissions
            \ENDFOR
        \ENDIF
        \IF{$getRole(e)$ is not Null}
            \STATE Add the roles to the ASRT (Active Subject Role Table) of $se$
            \STATE Inform $se$ of the permissions associated
        \ENDIF
    \ENDIF
    \end{algorithmic}
    \end{algorithm}

\subsection{FEAC Policy Specification}\label{ps}
\subsubsection{Emergency-role and Subject Selection}\label{erss}
In FEAC, emergency-role is temporarily designated to execute response actions which are not allowed to operate by normal-role. In emergency state, the subject can be associated with only one emergency-role because one subject can process only one emergency at a time; otherwise, the time waiting for the subject to finish performing the response actions will delay the process of other emergencies, potentially causing the emergencies uncontrollable. In addition, associating one subject with just one emergency-role is convenient for parallel process of the emergencies. The minimum permissions are assigned to the emergency-role in order to prevent damage from the execution of malicious entities.

The method of emergency-role and role-constraint mapping is used to hierarchically select the most suitable subjects to execute the response actions. In emergency-role mapping, each emergency-role is mapped to normal-roles (one to many), and thus a hierarchy structure is formed. The subjects of normal-roles in the higher hierarchy are more suitable than the ones in lower hierarchy as the subjects to execute response actions. The constraints of the normal-role in the emergency-role mapping are used to guarantee the correctness of the subject selection, such as the distance between the subject location, the place where the emergency occurred, the number of the subjects that can be associated with the emergency-role and the properties of subject (experience, licenses and certificates).

\subsubsection{Enabling Response Actions}\label{enableAct}
When the response actions and subjects have been selected, the next stage is to enable the permissions for the subjects to execute the response actions, to notify the subjects and to rescind the permissions when necessary. It includes four steps:
\begin{enumerate}
  \item According to the RL and the selected response actions, assign the corresponding permissions to the emergency-role. The Ed of the emergency is associated with the permissions to limit the execution time for the subjects to execute the response actions.
  \item Proactively alternate the emergency-role to the selected subjects, and record the normal-roles for recovery.
  \item Inform the chosen subjects to use the permissions to execute the response actions for eliminating the emergency.
  \item Rescind the assigned permissions after the emergency has been solved or the Ed has expired.
\end{enumerate}

\subsubsection{Policy Implementation}\label{pi}
Given the system structure of FEAC, this subsection will describe the principal components of the access control model and the policy implementation. Table~\ref{table1} shows the set of the components in the access control model, such as the notion of role, subject, object, permission, ACL, emergency and constraint. Table~\ref{table2} lists the table components of the access control model. SRT and ASRT store the roles that can be allocated and activated for subjects. TDT and EDT are the tables of time-dpendency and environment-dpendency. RMT and RCT are used to select the most appropriate subjects to execute the response actions. The original roles of the subject are stored in the ORT, and will be reassigned when necessary. The entities have the same function are grouped in EFGT.

\begin{table}
\caption{Set Components of Access Control Model}\label{table1}
\begin{tabular}{m{55pt}|m{130pt}}
\toprule
\multicolumn{2}{c}{Sets}\\
\midrule
Subject (S) & $S=set~of\{<Sid,Pt>\}$, where $Pt$ is the properties of the subjects, $Sid=unique~<string>$\\\hline
Object (O) & $O=set~of\{<Oid,ACL>\}$, where $Oid=unique~<string>$ and ACL is the access control list\\\hline
Role (R) & $S=NR \bigcup ER$, where $NR$ and $ER$ are both $set~of~\{<role>\}$, present the normal-role and emergency-role respectively\\\hline
Permission (P)& $P=set~of\{<Oid,op,td>\}$, where $td$ is the time duration limited the use of permission\\\hline
Access Control List (ACL) & $ACL=set~of\{<r,p>\}$, where $r\in R,p\in P$\\\hline
Emergency (E)& $E=set~of\{<Eid,TSs,Ed,Pro,e>\}$, where $Eid=uniqur~<string>,~TSs=set~of\{<TS>\},~Pro$ is the priority, and $e$ is the entity on which the emergency happens\\\hline
Constraint (C)& $C=set~of{<constraint>}$\\
\bottomrule
\end{tabular}
\end{table}
\begin{table}
\caption{Table Components of Access Control Model}\label{table2}
\begin{tabular}{m{60pt}|m{130pt}}
\toprule
\multicolumn{2}{c}{Tables}\\
\midrule
Subject Role Table (SRT)& $SRT=set~of\{<s,rs>\}$, where $s\in S,~\forall r\in rs,~r\in R$\\\hline
Active Subject Role Table (ASRT)& $ASRT=set~of\{<s,rs>\}$, where $ASRT\subset SRT$\\\hline
Time Dependency Table (TDT)& $TDT=set~of\{<Eid_{1},Eid_{2}>\}$, where emergency with $Eid_{1}$ executes before the one with $Eid_{2}$\\\hline
Environment Dependency Table (EDT)& $EDT=set~of~\{<e,Eid>\}$, entity $e$ dependents on the environment emergency with $Eid$\\\hline
Role Mapping Table (RMT)& $RMT=set~of\{<er,X,cr>\}$, where $\forall g\in X,~g\in NR,~cr\in C$\\\hline
Role Constraint Table (RCT)& $RCT=set~of\{<er,c>\}$, where $er\in ER,~c\in C$\\\hline
Old Role Table (ORT)& $ORT=set~of\{<s,rs>\}$, where $s\in S,~rs\subset ASRT$\\\hline
Entity Function Group Table (EFGT)& $EFGT=set~of~\{<e,FGid>\}$, where $FGid=unique~<string>$\\
\bottomrule
\end{tabular}
\end{table}


    \begin{algorithm}
    \caption{FEAC execution model}
    \label{alg3}
    \begin{algorithmic}[1]
    \STATE $selSubject\leftarrow NULL$
    \STATE $Mode\leftarrow Normal$
    \STATE $curState\leftarrow N$
    \WHILE{$true$}
        \STATE $t\leftarrow checkSysState()$
        \IF{$t\neq curState$}
            \IF{$t = N$}
                \STATE $Mode\leftarrow Normal$
            \ELSE
                \STATE $Mode\leftarrow Emergency$
            \ENDIF
            \STATE $RescPerm\leftarrow (t\otimes curState)$
            \STATE $curState\leftarrow t$
        \ENDIF
        \IF{$Mode = Emergency$}
            \FORALL{$emergencyGroup$}
                \IF{$hasNewEmry(egid)\neq NULL$}
                    \STATE $PDAGM(egid)$
                    \STATE $pv\leftarrow getPvalue(egid)$
                    \IF{$pv = NULL$}
                        \STATE $faultTolerance(egid)$
                        \IF{$prob\neq 1$}
                            \STATE $probFirstSel(egid)$
                        \ELSE
                            \STATE $timeFirstSel(egid)$
                        \ENDIF
                    \ELSE
                        \STATE $optimalSel(egid)$
                    \ENDIF
                \ENDIF
            \ENDFOR
            \STATE $curEmy\leftarrow findEmy(t)$
            \IF{$curEmy\neq NULL$}
                \FORALL{$c\in curEmy$}
                    \STATE $er\leftarrow emyRole(c)$
                    \STATE $selSubject\leftarrow roleMapSel(er)$
                    \IF{$selSubject = NULL$}
                        \STATE $selSubject\leftarrow roleConsSel(er)$
                    \ENDIF
                    \IF{$selSubject\neq NULL$}
                        \STATE $TS\leftarrow getTS(c)$
                        \STATE $addACL(er,TS)$
                        \STATE $alterRole(selSubject,er)$
                        \STATE $informSub(selSubject,TS,er)$
                    \ENDIF
                \ENDFOR
                \STATE $recordAct()$
            \ENDIF
        \ENDIF
        \STATE $wait(tp)$
    \ENDWHILE
    \end{algorithmic}
    \end{algorithm}

Algorithm~\ref{alg3} illustrates the execution model of FEAC, which works in a loop manner to monitor the transition of system state. The function $checkSysState()$ checks the system state and returns the current system state. If a change is detected, then the system rescinds the pervious permissions and sets the appropriate system state. When the system is in emergency state, the system decides whether to use fault-tolerance according to the $P$-$value$. Functions $probFristSel()$, $timeFirstSel()$ and $optimalSel()$ are chosen in different situation. Function $findEmy()$ finds out the set of emergencies that need to be processed. For each of the emergencies in the set, the subject to execute response actions is selected by function $roleMapSel()$ and $roleConsSel()$ using RMT and RCT respectively. Once the subjects have been selected, the system assigns the corresponding permissions to the emergency-role and informs the subject to use the permissions. All the actions are recorded to ensure any malicious activity can be detected. After these actions, the system waits for $tp$ duration of time, and repeats the whole process.

\section{Validation of FEAC}\label{val}

FEAC should work in a proactive and adaptive manner in order to provide the $right$ set of permissions to the $right$ set of subjects at $right$ time. In this section, we prove that the FEAC model can meet the following properties/requirements: responsiveness, correctness, security, liveness and non-repudiation.

\begin{theorem}\label{T1}
\textbf{Responsiveness}. FEAC ensures that once an emergency occurs, the system will detect it in time, assigns proper permissions, and informs the selected subjects.
\end{theorem}

\begin{proof}\label{P1}
The FEAC model periodically detects the emergencies as shown in Algorithm~\ref{alg3} on line 5. After waiting for $tp$ duration of time, the system calls the function $checkSysState()$ to get current state and compares it with the previous one in order to detect new emergency and change in system state. Lines 35 through 44 are used for subject selection, permission association and subject notification. The role of the subject is alternated by function $altherRole()$ and the permissions are added into the ACL by function $addACL()$.
\end{proof}

\begin{theorem}\label{T2}
\textbf{Correctness}. Emergency-role can be designated and the response actions can be executed if and only if the system is under emergency state. Different path selection algorithms are processed in corresponding situations.
\end{theorem}

\begin{proof}\label{P2}
If the system is under emergency state, the system mode changes to emergency mode, and the value returned from function $findEmy()$ is not empty. Then the codes after line 33 in Algorithm~\ref{alg3} will be executed. The emergency-role will be designated and be assigned to the selected subject, then the response actions will be performed. On the other hand, if the emergency-role is designated, it means that the condition is true and the system is in emergency state. In this case, if the $P$-$value$ is not 0, the $optimalSel()$ algorithm (line 28) is selected. Otherwise, according to the value of $prob$, maximum probability path (line 23) or minimum execution time path (line 25) is selected respectively.
\end{proof}

\begin{theorem}\label{T3}
\textbf{Security}. The permissions assigned to execute the response actions can only be used in emergency mode. Fault-tolerance can only be processed when emergency management fails.
\end{theorem}

\begin{proof}\label{P3}
When the system moves from emergency mode to normal mode, the function $rescPerm()$ will be executed. The permissions and the subject's role will also be rescinded. The permissions for executing the response actions can never be used in normal state. When the $P$-$value$ is 0 (line 20), function $faultTolerance()$ is called.
\end{proof}

\begin{theorem}\label{T4}
\textbf{Liveness.} The time duration of the access permissions are limited and the system must have the ability to rescind the permissions when the emergency is eliminated or the Ed is expired.
\end{theorem}

\begin{proof}\label{P4}
When the system state goes back to normal state, the access permissions are rescinded. On lines 8 and 12 in Algorithm~\ref{alg3}, the change of system state is detected, then on line 12 the permissions are rescinded to satisfy security requirements.
\end{proof}

\begin{theorem}\label{T5}
\textbf{Non-Repudiation}. Malicious actions performed in emergency situations are restricted and the subject cannot be repudiated.
\end{theorem}

\begin{proof}\label{P5}
The permissions for executing the response actions are limited by the Ed of emergency. On line 47  in Algorithm~\ref{alg3}, the $recordAct()$ function records all the actions performed in the system, such as the assignment of emergency-role and permissions, the execution of the response actions and the notification of the selected subjects. Once the malicious actions occur, the system will record them into the log files.
\end{proof}

\begin{table*}
\caption{Emergencies in hospital medical care}\label{tab3}
\begin{footnotesize}
\begin{tabular}{clclm{220pt}ll}
\toprule
ID& Emergency& prio& Ed&  Task-Set (TS)& Exec.Time& prob\\
\midrule
E1 & Fire& 3& 20 min& $\{<ICU Door, <u>>,<FireExtinguisher,<u>>\}$& 3 min& 0.8\\
E2 & Dust and Smoke& 4& 10 min& $\{<Ventilation fan, <u>>,<ICU Door,<u>>\}$& 2 min& 0.85\\
E3 & Cardiac Arrest& 6& 8 min& $\{<P1 Health Data, <r\&w>>,<Defibrillator1,<u>>,<ICU Door,<u>>\}$& 1 min& 0.8\\
E4 & Headache& 9& 30 min& $\{<P1 Health Data, <r\&w>>,<Medicine Room Door,<u>>,<ICU Door,<u>>\}$& 1 min& 0.9\\
E5 & Fever& 9& 20 min& $\{<P1 Health Data, <r\&w>>,<Medicine Room Door,<u>>,<ICU Door,<u>>\}$& 2 min& 0.95\\
E6 & Arrhythmia& 7& 18 min& $\{<P2 Health Data, <r\&w>>,<Electrocardiograph1,<u>>,<ICU Door,<u>>\}$& 2 min& 0.85\\
E7 & Angina Cordis& 8& 12 min& $\{<P2 Health Data, <r\&w>>,<ICU Door,<u>>\}$& 1 min& 0.9\\
\bottomrule
\end{tabular}
\end{footnotesize}
\end{table*}

\begin{figure*}
 	\centerline{\includegraphics[width=6in]{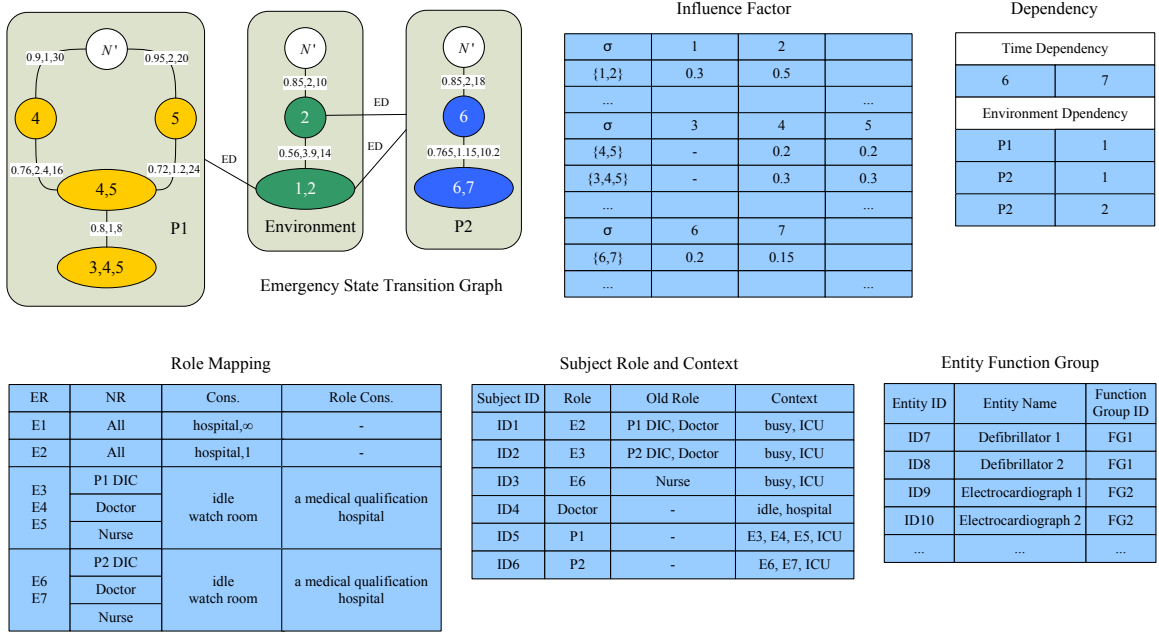}}
  \caption{PD-AGM for hospital medical care}\label{fig7}
\end{figure*}

\section{Case Study}\label{cs}

In this section, we present an example to illustrate how FEAC can be used for CPS applications. The case of hospital medical care shows the ability of FEAC to handle multiple emergencies in emergency situations. The emergencies which occurred in hospital medical care application are shown in Table~\ref{tab3}. $E1$ and $E2$ are environment emergencies. $E3$ to $E5$ are emergencies that occurred on patient 1. $E6$ and $E7$ are emergencies that occurred on patient 2. Figure~\ref{fig7} shows the PD-AGM for the hospital medical care application. The generation of the emergency state transition graph uses the dependency relationships and influence factor in Figure~\ref{fig7}. To simplify calculation, the values of $\alpha$ and $\beta$ are set to 1. The dependency between entity $P1$ and emergency $E1$ is environment dependency. Environment dependency also exists between entity $P2$ and $E1$, as well as entity $P2$ and $E2$. The emergencies in the emergency-groups of $P1$ and $P2$ must wait the accomplishment of the corresponding environment emergencies. Notice that, the two path of state \{4, 5\} have the same $P$-$value$ of 0.684, and the corresponding execution time are 3.4min and 3.2min respectively. Then the right RL path is selected.

The principal components of the FEAC for the hospital medical care application are shown in Figure~\ref{fig7}. All the normal-roles map to the emergency-roles $E1$ and $E2$. Constraints for selecting subjects include the position limitation and the number of execution. The role-constraints are used when no subject has been selected. The original roles of the subject are recorded, and reassigned when the response actions have been executed. The entries in ACL of the objects influence the access control decisions for the access request of specific subject.

\section{Conclusion}\label{con}

In this paper, a fault-tolerant emergency-aware access control scheme called FEAC has been presented, which provides proactive and adaptive access control polices to address the multiple emergencies management problem and fault-tolerance problem for CPS applications. The PD-AGM model is introduced to select the optimal response action path for eliminating all the active emergencies. The priority and dependency relationships of emergencies are used to exclude the infeasible response action paths and relieve the emergency combination state explosion problem. In order to handle all the emergencies timely, emergency-group and emergency-role are introduced for processing multiple emergencies in parallel. FEAC can meet responsiveness, correctness, security, liveness and non-repudiation requirements. A case study of hospital medical care application has illustrated the effectiveness of FEAC.

 \section*{Acknowledgments}
 This work was partially supported by the National Natural Science Foundation of China under Grant No.60703101 and No. 60903153 and the Fundamental Research Funds for the Central Universities.


\end{document}